\documentclass{article}
\usepackage{amsthm}
\usepackage{amssymb}
\usepackage{amsfonts}
\begin{document}

\newtheorem{theorem}{Theorem}[section]
\newtheorem{definition}{Definition}[section]
\newtheorem{corollary}[theorem]{Corollary}
\newtheorem{lemma}[theorem]{Lemma}
\newtheorem{proposition}[theorem]{Proposition}
\newtheorem{step}[theorem]{Step}
\newtheorem{example}[theorem]{Example}
\newtheorem{remark}[theorem]{Remark}

\font\sixbb=msbm6
\font\eightbb=msbm8
\font\twelvebb=msbm10 scaled 1095
\newfam\bbfam
\textfont\bbfam=\twelvebb \scriptfont\bbfam=\eightbb
                           \scriptscriptfont\bbfam=\sixbb

\newcommand{\tr}{{\rm tr \,}}
\newcommand{\linspan}{{\rm span\,}}
\newcommand{\rank}{{\rm rank\,}}
\newcommand{\diag}{{\rm Diag\,}}
\newcommand{\Image}{{\rm Im\,}}
\newcommand{\Ker}{{\rm Ker\,}}

\def\bb{\fam\bbfam\twelvebb}
\newcommand{\enp}{\begin{flushright} $\Box$ \end{flushright}}
\def\cD{{\mathcal{D}}}
\def\N{\bb N}
\def\R{\bb R}
\def\Z{\bb Z}
\def\C{\bb C}

\title{Continuous coexistency preservers on effect algebras 
\thanks{The first author was supported by Leading Graduate Course for Frontiers of Mathematical Sciences and Physics (FMSP), and JSPS KAKENHI 19J14689, Japan.}
\thanks{The second author was supported by grants N1-0061, J1-8133, and P1-0288 from ARRS, Slovenia.}}
\author{Michiya Mori\footnote{Graduate School of Mathematical Sciences, The University of Tokyo, 3-8-1 Komaba Meguro-ku Tokyo, 153-8914, Japan, mmori@ms.u-tokyo.ac.jp} \ \, and \ \, 
Peter \v Semrl\footnote{Faculty of Mathematics and Physics, University of Ljubljana,
        Jadranska 19, SI-1000 Ljubljana, Slovenia; Institute of Mathematics, Physics, and Mechanics, Jadranska 19, SI-1000 Ljubljana, Slovenia, peter.semrl@fmf.uni-lj.si}
        }

\date{}
\maketitle

\begin{abstract}
Let $H$ be a finite-dimensional Hilbert space, $\dim H \ge 2$. We prove that every continuous coexistency preserving map on the effect algebra $E(H)$ is either a standard automorphism
of $E(H)$, or a standard automorphism of $E(H)$ composed with the orthocomplementation. We present examples showing the optimality of the result.
\end{abstract}
\maketitle




\section{Introduction and statement of the main result}

One of the main objectives of quantum mechanics is  the study of measurement.
In the  mathematical formulation of quantum mechanics an observable is represented by a self-adjoint operator. 
However, such a representation implicitly assumes that measurements are perfectly accurate, which cannot be true.
This led G\"unther Ludwig to introduce an alternative axiomatic formulation of quantum mechanics, see \cite{Kraus, LudI, LudII}.
The main difference compared to the classical approach is that quantum events are not sharp, and
therefore, according to Ludwig, a quantum event is not necessarily a projection, but rather a self-adjoint operator whose spectrum lies in $[0,1]$.
Such operators are called effects and the set of all effects is called the effect algebra. 
In this theory one of the most important relations is coexistency. Ludwig defined two effects to be coexistent if they can be measured together by applying a suitable apparatus.
For more details we refer to \cite{GeS} and the references therein.

In the language of mathematics (see \cite{LudI}), effects, the effect algebra, and the relation of coexistency are defined in the following way. 
Let $H$ be a Hilbert space. An effect is a bounded linear self-adjoint operator $A : H \to H$ satisfying $0 \le A \le I$. The set of all effects
will be denoted by $E(H)$. We further denote by $P(H) \subset E(H)$ the set of all projections (bounded linear self-adjoint idempotent operators on $H$),
by $P^1 (H) \subset P(H)$ the set of all projections of rank one, and by ${\rm Sca}\, (H)$ 
the set of all scalar effects,
${\rm Sca}\, (H) = \{ tI \, : \, t \in [0,1] \}$.
For each $A \in E(H)$ its orthocomplement  is defined by $A^\perp = I - A$.
In the case when $H$ is finite-dimensional, $\dim H = n$, we identify bounded self-adjoint operators on $H$ with hermitian $n\times n$ matrices. With this identification we
have $E(H) = E_n$, $P(H) = P_n$, and $P^1 (H) = P_{n}^1$, where $E_n$ is the set of all $n\times n$ hermitian matrices whose spectra belong to the unit interval $[0,1]$, $P_n$
is the set of all $n \times n$ hermitian matrices $P$ satisfying $P^2 = P$, and $P_{n}^1$ is the set of all $n\times n$ hermitian rank one matrices $P$ satisfying $P^2 = P$.

Two effects $A,B \in E(H)$ are said to be coexistent, $A \sim B$, if there exist effects $E,F,G \in E(H)$ such that
$$
A = E + G  \ \ \ {\rm and} \ \ \ B = F + G \ \ \ {\rm and} \ \ \ E + F + G \in E(H).
$$
We say that a map $\phi : E(H) \to E(H)$ preserves coexistency if for every pair $A,B \in E(H)$ we have
\begin{equation}\label{ccc}
A \sim B \iff \phi (A) \sim \phi (B).
\end{equation}
The relation of coexistence is very poorly understood even for qubit effects, that is, elements of $E_2$. An attempt of the description of coexistence on $E_2$ can be found 
in \cite{BuS} but there are no similar results in higher dimensions.

In mathematical foundations of quantum mechanics symmetries play a very important role. These are bijective maps on various quantum structures that preserve certain
relations and/or operations that are relevant in mathematical physics. We refer to \cite{GeS,Mol0,Se0,Se1,Se2} and the references therein for more information on symmetries
of effect algebras. It turns out that quite often these symmetries are standard automorphisms, that is, maps of the form $A \mapsto UAU^\ast$, $A\in E(H)$,
where $U : H \to H$ is a unitary or antiunitary operator. 

One of the most challenging problem in the theory of symmetries of effect algebras, that is, the problem of describing the general form of bijective coexistency preservers
on $E(H)$, has been resolved only very recently. It was proved by Geher and \v Semrl  \cite{GeS} that for every bijective map $\phi : E(H) \to E(H)$ satisfying (\ref{ccc})
there exists a unitary or antiunitary operator $U: H\to H$ and a bijective map $g : [0,1] \to [0,1]$ such that 
$$
\{\phi(A), \phi(A^\perp)\} = \{UAU^\ast, UA^\perp U^\ast\} 
$$	
for every $A\in E(H) \setminus {\rm Sca}\, (H)$, 
and $$
\phi (tI) = g(t)I 
$$
for every real $t \in [0,1]$. Conversely, every map of the above form preserves coexistency in both directions.
Of course, such maps are far from being continuous in general. Under the additional assumption of continuity we get a much nicer conclusion 
that every such map must be a standard automorphism of $E(H)$ or a standard automorphism composed with the orthocomplementation. More precisely,
for every bijective continuous map $\phi : E(H) \to E(H)$ satisfying (\ref{ccc})
there exists a unitary or antiunitary operator $U: H\to H$ such that either
$$
\phi (A) = UAU^\ast, \ \ \ A \in E(H),
$$
or
$$
\phi (A) = UA^\perp U^\ast, \ \ \ A \in E(H).
$$
As we have added the continuity assumption it is natural to ask if other assumptions in the result of Geher and \v Semrl
can be weakened. In particular, can we remove the bijectivity assumption and still get a nice
description of coexistency preservers under the additional assumption of continuity? The answer is negative. To see this
we assume that $H$ is an infinite-dimensional Hilbert space. Let further $T: H \to H$ be a bounded linear contraction, that is, $\| T \| \le 1$.
Then the operator given by  $\varphi_1 (A) = TAT^\ast$, $A\in E(H)$, maps $E(H)$ into itself. Indeed, for every $A\in E(H)$ the linear map 
$TAT^\ast$ is a bounded self-adjoint positive operator (throughout the paper a positive operator/matrix will mean a positive-semidefinite operator/matrix). It is also trivial to verify that $TAT^\ast \le I$. Let $\varphi_2 : E(H) \to S \subset E(H)$ be any
continuous map where $S = \{ A \in E(H) \, : \, \| A \| \le 1/2\}$. Finally, let $G_n$, $n=1,2,\ldots$, be pairwise commuting effects.
We denote by $S(H)$ the real linear space of all bounded self-adjoint operators on $H$. Let $f_n : S(H) \to \R$, $n=1,2,\ldots$, be any positive real-linear
functionals of norm one. We define $\varphi_3 : E(H) \to E(H)$ by $\varphi_3 (A) = \sum_{j=1}^\infty 2^{-j} f_j (A) G_j$. Since $H$ is
infinite-dimensional it can be identified with the direct orthogonal sum of four copies of $H$. Hence, each map from $E(H)$ into $E( H \oplus H \oplus H \oplus H)$
can be considered as a map from $E(H)$ into itself.

\begin{proposition}\label{fitne}
The map $\phi : E(H) \to E( H \oplus H \oplus H \oplus H)$ given by
$$
\phi (A) = \left[ \begin{matrix} {
A & 0 & 0 & 0 \cr 0 & \varphi_1 (A) & 0 & 0 \cr 0 & 0 & \varphi_2 (A) & 0 \cr 0 & 0 & 0 & \varphi_3 (A) \cr
}\end{matrix}\right] , \ \ \ A\in E(H),
$$
is continuous and satisfies (\ref{ccc}).
\end{proposition}

It is obvious that $\phi (A) \in  E( H \oplus H \oplus H \oplus H)$ for every $A\in E(H)$ and that $\phi$ is continuous. We postpone the verification of (\ref{ccc}) to the next section.
The above example shows that there is no nice description of continuous coexistency preservers on $E(H)$. However, a careful reader has noticed that the assumption that
$H$ is infinite-dimensional was essential in constructing the above counterexample.
And in fact, as our main theorem shows, in the finite-dimensional case the answer to our question is in the affirmative.

\begin{theorem}\label{main}
Let $H$ be a Hilbert space, $2 \le \dim H < \infty$. Assume that $\phi : E(H) \to E(H)$ is a continuous map satisfying (\ref{ccc}).
Then there exists a unitary or antiunitary operator $U : H \to H$ such that either
$$
\phi (A) = UAU^\ast
$$
for every $A\in E(H)$, or
$$
\phi (A) = UA^\perp U^\ast
$$
for every $A \in E(H)$.
\end{theorem}

In the language of linear algebra this statement reads as follows. Let $n \ge 2$ be an integer and $\phi : E_n \to E_n$ a continuous map. Assume that for every
pair of effects $A,B \in E_n$ we have $A \sim B \iff \phi (A) \sim \phi (B)$. Then there exists a unitary $n \times n$ matrix $U$ such that either $\phi (A) = UAU^\ast$ for all $A\in E_n$, or
$\phi (A) = UA^t U^\ast$ for all $A\in E_n$, or $\phi (A) = U(I-A)U^\ast$ for all $A\in E_n$, or $\phi (A) = U(I-A^t) U^\ast$ for all $A\in E_n$. Here, $A^t$ denotes the transpose of the
matrix $A$. 

Proposition \ref{fitne} shows that the assumption of finite-dimensionality is indispensable in our main result. In order to show that no further improvements are possible we need
to consider one more problem. 
We will say that a map $\phi : E(H) \to E(H)$ preserves coexistency \sl in one direction only \rm if for every pair $A,B \in E(H)$ we have
\begin{equation}\label{ccc1}
A \sim B \Rightarrow \phi (A) \sim \phi (B).
\end{equation}
The question is, of course, whether the assumption of preserving coexistency in our main theorem can be replaced by the weaker 
 assumption of preserving coexistency in one direction only and still get the same conclusion. We will see that the answer is negative even if we add the bijectivity assumption, thus confirming the
optimality of Theorem \ref{main}.

In fact, we will show even more. Recall first that the basic theorem concerning symmetries of effect algebras is Ludwig's description of ortho-order automorphisms. His result states that
every bijective map $\phi : E(H) \to E(H)$, $\dim H \ge 2$, such that for every pair $A,B \in E(H)$ we have
$$
A \le B \iff \phi (A) \le \phi (B)
$$
and 
$$
\phi (A^\perp ) = \phi (A)^\perp,
$$
is a standard automorphism of $E(H)$. In \cite{Se2} the optimality of Ludwig's theorem was studied. It was shown that there exists a 
 bijective map $\phi : E(H) \to E(H)$ such that for every pair $A,B \in E(H)$ we have (\ref{ccc1}),
\begin{equation}\label{mujko}
A \le B \Rightarrow \phi (A) \le \phi (B),
\end{equation}
and 
\begin{equation}\label{mujko2}
\phi (A^\perp ) = \phi (A)^\perp ,
\end{equation}
which is not a standard automorphism of $E(H)$. However, the example was non-continuous. Here, we will provide an example of a continuous map with all the above properties. Moreover, the
example is even simpler than the one presented in \cite{Se2}.

\begin{proposition}\label{micmo}
Take any continuous monotone increasing function $f : [0,1] \to [0,1]$ with $f(t) >0$ for all $t > 0$ and $f(1) = 1$.
The map $\phi : E_n \to E_n$ given by
$$
\phi (A) = f( {\rm tr}\, A) A \ \ \ {\rm if} \ \, 0 \le {\rm tr}\, A \le 1,
$$
(here, ${\rm tr}\, A$ denotes the trace of $A$)
$$
\phi (A) = A  \ \ \ {\rm if} \ \, 1 \le {\rm tr}\, A \le n-1,
$$
and
$$
\phi (A) = \phi ( A^\perp )^\perp  \ \ \ {\rm if} \ \, n-1 \le {\rm tr}\, A \le n,
$$
is bijective, continuous and satisfies (\ref{ccc1}), (\ref{mujko}), and (\ref{mujko2}).
\end{proposition}
This example came as a surprise. Our expectation was that in the presence of the assumptions of bijectivity, continuity, and finite-dimensionality we will be able to prove that every map
preserving enough properties in one direction only has to be a standard automorphism. And in fact, we were able to prove quite a lot of nice structural properties of such maps but
instead of coming to the desired conclusion the detailed analysis of such maps brought us to the above example showing that our starting conjecture was wrong.

Our first strategy to prove the main theorem was to use topological tools to show that $\phi$ is bijective and then to apply the known result on bijective maps preserving coexistency on $E(H)$.
It turned out that a shorter direct proof is possible. But of course, we are still using quite a few ideas from \cite{GeS}. In the next section we will first formulate few lemmas that have been
proved before. In particular, Propositions \ref{fitne} and \ref{micmo} will be deduced from some of the known results.  
In the rest of the paper all the ideas used in the proofs are new. After proving a few new technical results in the second section we will give the
proof of our main theorem in the last section.

\section{Preliminary results}

Let $H$ be a Hilbert space, $\dim H \ge 2$.
We need some more notation. For $A\in E(H)$ we denote
$$
A^\sim = \{ B \in E(H)\, : \, A \sim B \}
$$
and 
$$
A^c = \{ B \in E(H)\, : \, AB = BA \}.
$$
The following facts about coexistency are well-known (see \cite[p.440]{Mol0} and \cite[p.140]{Mol}):

\begin{lemma}\label{properties}
For every $A,B \in E(H)$ and $P\in P(H)$ we have:
\begin{itemize}
		\item $A^\sim = E(H)$ if and only if $A \in {\rm Sca}\, (H)$,
		\item $P^\sim = P^c$,
		\item $A^c \subset A^\sim$,
\item if $A$ and $B$ are rank one effects with different images then $A \sim B$ if and only if  $A+B \in E(H)$.
	\end{itemize}
\end{lemma}

The next two lemmas were proved in \cite{Se2}.

\begin{lemma}\label{dirsum}
Let $H = H_1 \oplus \ldots \oplus H_k$ be an orthogonal direct sum decomposition. Assume that effects $A,B \in E(H)$
have the following operator matrix representations with respect to this direct sum decomposition:
$$
A = \left[ \matrix{ A_1 & 0 & \ldots & 0 \cr 0 & A_2 & \ldots & 0 \cr  \vdots & \vdots & \ddots & \vdots \cr 0 & 0 & \ldots & A_k \cr} \right] \ \ \ {\rm and} \ \ \ 
B = \left[ \matrix{ B_1 & 0 & \ldots & 0 \cr 0 & B_2 & \ldots & 0 \cr  \vdots & \vdots & \ddots & \vdots \cr 0 & 0 & \ldots & B_k \cr} \right].
$$
Then $A_j \sim B_j$, $j=1, \ldots , k$, if and only if $A \sim B$.
\end{lemma}

\begin{lemma}\label{lem3}

	For any $A,B \in E(H)$ the following are equivalent:
	\begin{itemize}
		\item $A \sim B$,
		\item there exist effects $M,N \in E(H)$ such that $M \le A$, $N \le I-A$, and $M+N = B$.
	\end{itemize}
\end{lemma}

It is now easy to prove Proposition \ref{fitne}.

\begin{proof}[Proof of Proposition \ref{fitne}]
We only need to verify (\ref{ccc}). If $\phi (A) \sim \phi (B)$ then by Lemma \ref{dirsum} we have $A \sim B$. Applying this lemma once more we see that it remains to show that for every pair $A,B \in E(H)$ we have
$A \sim B \Rightarrow \varphi_j (A) \sim \varphi_j (B)$, $j=1,2,3$. 
In the case $j=1$ the desired implication follows easily from Lemma \ref{lem3}. 
The case $j=2$ is clear by the definition of coexistency.
In the case $j=3$ we apply the item three of Lemma \ref{properties} to conclude the proof.
\end{proof}

In order to prove Proposition \ref{micmo} we need the following result from \cite{GeS}.

\begin{lemma}\label{convex}
For every $A\in E(H)$ the set $A^\sim$ is convex.
\end{lemma}

\begin{proof}[Proof of Proposition \ref{micmo}]
Clearly, $\phi$ is continuous and satisfies
$$
\phi (A^\perp) = \phi(A)^\perp
$$
for every $A \in E(H)$.

To see that it is bijective we only need to verify that the restriction of $\phi$ to the set of all effects $A$ satisfying $0 \le {\rm tr}\, A \le 1$ is a bijection of this set onto itself. According to our assumptions
the function $t \mapsto f(t)t$ is a monotone increasing bijection of $[0,1]$ onto itself. If $0 \le {\rm tr}\, A \le 1$ and $0 \le {\rm tr}\, B \le 1$ and
$$
f({\rm tr}\, A) A = f({\rm tr}\, B) B,
$$
then $f({\rm tr}\, A) \, {\rm tr}\, A = f({\rm tr}\, B) \, {\rm tr}\, B$, and therefore ${\rm tr}\, A = {\rm tr}\, B$. It follows that $A=B$.
For any nonzero $B$ with $0 < {\rm tr}\, B = s \le 1$ there exists $t \in [0,1]$ with $f(t)t = s$. Set $A = (t/s)B$. Obviously, ${\rm tr}\, A = t \in [0,1]$. We have
$$
\phi (A) = f(t) (t/s) B = B,
$$
as desired.

Assume now that for $A,B \in E(H)$ we have $A \le B$. Then ${\rm tr}\, A \le {\rm tr}\, B$. Moreover, we have $\phi (A) \le A$ if $0 \le {\rm tr}\, A \le 1$ and
because $\phi$ is $\perp$-preserving this yields $\phi (A) \ge A$ if $n-1 \le {\rm tr}\, A \le n$. It is now straightforward to verify that $\phi (A) \le \phi (B)$.

Finally, suppose that the effects $A,B$ are coexistent. We must verify that then $\phi (A)$ and $\phi (B)$ are coexistent, too. Since $\phi$ is $\perp$-preserving map there is no loss of
generality in assuming that ${\rm tr}\, A , {\rm tr}\, B \le n/2$. Since both $A$ and $0$ are coexistent with $B$ and $B^\sim$ is convex, $cA$ is coexistent with $B$ for any $c\in [0, 1]$,
and therefore, $\phi (A)$ is coexistent with $B$. Using exactly the same arguments we conclude that $\phi (B)$ is coexistent with $\phi (A)$.
\end{proof}

A trivial consequence of Lemma \ref{lem3} is that for every $A\in E(H)$ we have $A^\sim = (A^\perp )^\sim$. But actually, much more is true. The following statement is one of the two main results in \cite{GeS}.

\begin{lemma}\label{lem4}
For any pair $A,B \in E(H)$ the following are equivalent:
\begin{itemize}
		\item $A^\sim = B^\sim$,
		\item $A= B$ or $A= B^\perp$ or $A,B \in {\rm Sca}\, (H)$.
	\end{itemize}
\end{lemma}

Till the end of this section $H$ will denote a finite-dimensional Hilbert space, $\dim H = n \ge 2$.
Let $p$ and $q$ be nonnegative integers, $p+q \le n$.
Then $E(p,q) \subset E_n$ is defined to be the set of all $A\in E_n$ such that $1$ is an eigenvalue of $A$ with the multiplicity $p$ and
$0$ is an eigenvalue of $A$ with the multiplicity $q$. In particular, $E(0,0)$ is the set of all effects $A$ such that both $A$ and $A^\perp$
are invertible and $E(p, n-p)$ is the set of all projections of rank $p$. 
Recall that $E(p, n-p)$ can be identified with the Grassmann space of all $p$-dimensional subspaces of ${\C}^n$, and therefore, $E(p,n-p)$ is a compact connected
topological manifold without boundary. 
Each $A\in E(p,q)$ is unitarily similar to a block diagonal matrix
\begin{equation}\label{ljil}
\left[ \begin{matrix} { I_p & 0 & 0 \cr 0 & B & 0 \cr 0 & 0 & 0_q \cr} \end{matrix} \right] ,
\end{equation}
where $I_p$ is the $p \times p$ identity matrix, $0_q$ is the $q\times q$ zero matrix and $B$ is an $(n-p-q)\times (n-p-q)$ diagonal matrix whose all
diagonal entries belong to the open interval $(0,1)$.

\begin{lemma}\label{lem5}
Let $p$ and $q$ be nonnegative integers, $p+q \le n$. If $A\in E(p,q)$ and a subset $U \subset A^\sim$ is homeomorphic to ${\R}^k$ for some positive integer $k$, then $k \le n^2 - 2pq$.
\end{lemma}

\begin{proof}
Without loss of generality we can assume that $A$ is of the form (\ref{ljil}). If $B \in E(H)$ belongs to $A^\sim$, then by Lemma \ref{lem3} we can find effects $M,N$ such that
$B = M + N$ and
$$
M \le \left[ \begin{matrix} { I_p & 0 & 0 \cr 0 & B & 0 \cr 0 & 0 & 0_q \cr} \end{matrix} \right] \ \ \ {\rm and} \ \ \
N \le \left[ \begin{matrix} { 0_p & 0 & 0 \cr 0 & B^\perp & 0 \cr 0 & 0 & I_q \cr} \end{matrix} \right] ,
$$
and therefore $B$ is a matrix of the form
\begin{equation}\label{aram}
\left[ \begin{matrix} { * & * & 0 \cr * & * & * \cr 0 & * & * \cr} \end{matrix} \right] .
\end{equation}
The real vector space of all hermitian matrices of the form (\ref{aram}) is of dimension $n^2 - 2pq$. Hence, $A^\sim$ can be considered as a subset of ${\R}^{n^2 - 2pq}$ and
the desired conclusion follows from the invariance of domain theorem.
\end{proof}

\begin{lemma}\label{restr}
Let $A,B$ be positive $n \times n$ hermitian matrices such that $0 \le B \le A$. Assume further that $A$ is invertible. Let $\varepsilon$ be any positive real number.
Then there exists a positive $n \times n$ hermitian matrix $C$ such that $0 \le C \le A$, both $C$ and $A-C$ are invertible, and $\| B-C \| < \varepsilon$.
\end{lemma}

\begin{proof}
With no loss of generality we may assume that $A=I$. Indeed, if $A\not= I$, then we may replace $B$ and $A$ by $A^{-1/2}B A^{-1/2}$ and $A^{-1/2} A A^{-1/2}$, respectively. 
In the next step we apply the unitary similarity to conclude that there is no loss of generality if we further assume that $B$ is a diagonal matrix whose diagonal entries belong to the unit interval $[0,1]$.
We can get the matrix $C$ with the desired properties by an arbitrarily small perturbation of diagonal entries of $B$.
\end{proof}

Let $A,B$ be $n \times n$ hermitian matrices. We will write $A < B$ if $A \le B$ and $B-A$ is invertible.

\begin{lemma}\label{lenov}
Let $p$ and $q$ be nonnegative integers, $p+q \le n$. Assume that $A\in E(p,q)$ and that $U\subset E(0,0)$ is an open subset such that
$$
 U \cap  A^\sim \not= \emptyset.
$$
Then there exists a subset $W \subset E(0,0)$  such that $W \subset A^\sim$, $W \subset U$, and $W$ is homeomorphic to ${\R}^{n^2 - 2pq}$. 
\end{lemma}

\begin{proof}
We may assume that $A$ is of the form (\ref{ljil}). Using $U \cap  A^\sim \not= \emptyset$
and Lemma \ref{lem3} we can find effects $M,N$ such that
$M + N \in U$ and
$$
M \le \left[ \begin{matrix} { I_p & 0 & 0 \cr 0 & B & 0 \cr 0 & 0 & 0_q \cr} \end{matrix} \right] \ \ \ {\rm and} \ \ \
N \le \left[ \begin{matrix} { 0_p & 0 & 0 \cr 0 & B^\perp & 0 \cr 0 & 0 & I_q \cr} \end{matrix} \right] .
$$
It follows that matrices $M$ and $N$ are of the form
$$
\left[ \begin{matrix} { * & * & 0 \cr * & * & 0 \cr 0 & 0 & 0 \cr} \end{matrix} \right] \ \ \ {\rm and} \ \ \
\left[ \begin{matrix} { 0 & 0 & 0 \cr 0 & * & * \cr 0 & * & * \cr} \end{matrix} \right],
$$
respectively. Applying Lemma \ref{restr} and the fact that $U$ is open we can find effects $K,L$ such that
$K+L \in U$, 
$$
K = \left[ \begin{matrix} { K_1 & K_2 & 0 \cr K_3 & K_4 & 0 \cr 0 & 0 & 0 \cr} \end{matrix} \right] \ \ \ {\rm and} \ \ \
L = \left[ \begin{matrix} { 0 & 0 & 0 \cr 0 & L_1 & L_2 \cr 0 & L_3 & L_4 \cr} \end{matrix} \right] ,
$$
and
$$
0 < \left[ \begin{matrix} { K_1 & K_2  \cr K_3 & K_4  \cr} \end{matrix} \right] <  \left[ \begin{matrix} { I_p & 0  \cr 0 & B  \cr} \end{matrix} \right]
\ \ \ {\rm and} \ \ \
0 < \left[ \begin{matrix} {  L_1 & L_2 \cr  L_3 & L_4 \cr} \end{matrix} \right] < \left[ \begin{matrix} { B^\perp & 0  \cr 0 & I_q  \cr} \end{matrix} \right] .
$$
Because $U$ is open we can find an open subset $W_1$ of the set of all $(n-q) \times (n-q)$ hermitian matrices and
an open subset $W_2$ of the set of all $(n-p) \times (n-p)$ hermitian matrices such that
$$
\left[ \begin{matrix} { K_1 & K_2  \cr K_3 & K_4  \cr} \end{matrix} \right] \in W_1 \ \ \ {\rm and} \ \ \ 
\left[ \begin{matrix} {  L_1 & L_2 \cr  L_3 & L_4 \cr} \end{matrix} \right]  \in W_2 ,
$$
\begin{equation}\label{aaq}
0 < T_1 <  \left[ \begin{matrix} { I_p & 0  \cr 0 & B  \cr} \end{matrix} \right]
\ \ \ {\rm and} \ \ \
0 < T_2 < \left[ \begin{matrix} { B^\perp & 0  \cr 0 & I_q  \cr} \end{matrix} \right] 
\end{equation}
for every $T_1 \in W_1$ and every $T_2 \in W_2$, and
$$
W_3 = \left\{ S \, : \, S= \left[ \begin{matrix} { S_1 & 0  \cr 0 & 0_q  \cr} \end{matrix} \right] + \left[ \begin{matrix} { 0_p & 0  \cr 0 & S_2  \cr} \end{matrix} \right]
\ \, S_1 \in W_1 \ \, {\rm and} \ \, S_2 \in W_2 \right\} \subset U.
$$
It follows from (\ref{aaq}) and Lemma \ref{lem3} that $W_3 \subset E(0,0) \cap A^\sim$. Hence, $W_3  \subset E(0,0) \cap A^\sim \cap U
$ is an open subset of the  $(n^2 - 2pq)$-dimensional real vector space of all $n \times n$ hermitian matrices
of the form
$$
\left[ \begin{matrix} { * & * & 0 \cr * & * & * \cr 0 & * & * \cr} \end{matrix} \right] ,
$$ 
and therefore, there exists a subset $W \subset W_3$ that is
homeomorphic to ${\R}^{n^2 - 2pq}$.
\end{proof}

\section{Proof of the main result}

We will divide the proof of Theorem \ref{main} into steps. So, assume that
$\dim H = n < \infty$, $n \ge 2$, and $\phi : E(H) \to E(H)$ is a continuous map satisfying (\ref{ccc}).

\begin{step}\label{prvi}
If $\phi (A) = \phi (B)$ for some $A,B \in E(H)$, then $A = B$ or $A= B^\perp$ or both $A$ and $B$ are scalar effects.
\end{step}

\begin{proof}
Let $C \in E(H)$ by any effect. Then $C \sim A$ if and only if $\phi (C) \sim \phi (A) = \phi (B)$ which is further equivalent to
$C \sim B$. Hence, $A^\sim = B^\sim$, and the conclusion follows from Lemma \ref{lem4}.
\end{proof}

\begin{step}\label{mart}
Let $p,q$ be nonnegative integers, $p + q \le n$. Then for every $A\in E(p,q)$ there exist
nonnegative integers $p', q'$, $p' + q' \le n$ such that
$$
p'q' = pq
$$
and
$$
\phi (A) \in E(p',q').
$$
\end{step}

\begin{proof} Let $A \in E(p,q)$ and $\phi(A)\in E(p', q')$.
With no loss of generality we can assume that $A$ is a diagonal matrix,
$$ 
A =  \left[ \begin{matrix} { I_p & 0 & 0  \cr 0 & D &0 \cr 0& 0 & 0_q  \cr} \end{matrix} \right] ,
$$
where $D$ is a diagonal matrix whose all eigenvalues belong to the open interval $(0,1)$. 
By Lemma \ref{lem3}, $A^{\sim}$ contains the set 
$$
\left\{ \left[ \begin{matrix} { M & 0 \cr 0 & 0_q \cr} \end{matrix} \right] + \left[ \begin{matrix} { 0_p & 0 \cr 0 & N \cr} \end{matrix} \right] \, : \, 
0_{n-q} < M < \left[ \begin{matrix} { \frac{1}{2}I_p & 0 \cr 0 & \frac{1}{2}D \cr} \end{matrix} \right], \ \
0_{n-p} < N < \left[ \begin{matrix} { \frac{1}{2}D^\perp & 0 \cr 0 & \frac{1}{2}I_q \cr} \end{matrix} \right] \right\}.   
$$
Let us denote the latter set ${\cal M}$. 
By Step \ref{prvi}, the set ${\cal M}\setminus {\rm Sca}\, (H)$ is injectively (and continuously) mapped into $\phi(A)^{\sim}$.  
Moreover, ${\cal M}\setminus {\rm Sca}\, (H)$ contains a subset that is homeomorphic to ${\R}^{n^2-2pq}$. 
We can also see that $\phi(A)^{\sim}$ is homeomorphic to a subset of ${\R}^{n^2-2p'q'}$. 
By the invariance of domain theorem, we obtain $n^2-2pq\le n^2-2p'q'$. 

Denote by $B$ the diagonal matrix with the diagonal entries $1/3, 1/4, \ldots, 1/(n+2)$. Since $B$ and $A$ commute we have $B \sim A$. Further, denote by ${\cal O} \subset E(0,0)$
the open set
$$
{\cal O} = \left\{ C \in E(H) \, : \, \| C - B \| < {1 \over n + 22} \right\}.
$$
Then it is clear that ${\cal O}$ is an open subset of $H_n$, the set of all $n\times n$ hermitian matrices. Moreover, ${\cal O}$ does not contain
any scalar effect and all effects contained in ${\cal O}$ have all eigenvalues $< 1/2$. Consequently, if $C \in {\cal O}$ then certainly $C^\perp \not\in {\cal O}$. It follows from Step \ref{prvi}
that $\phi$ maps ${\cal O}$ continuously and injectively into $H_n$ that  is homeomorphic to ${\R }^{n^2}$. By the invariance of domain theorem $\phi ({\cal O})$ is an open subset of $H_n$.
Consequently,  $\phi ({\cal O})$ is an open subset of $E(0,0)$ and the restriction of $\phi$ to ${\cal O}$ is a homeomorphism of ${\cal O}$ onto $\phi ({\cal O})$.

Since $\phi (B) \in 
\phi ( {\cal O})$ we have
$$
\phi ({\cal O}) \cap (\phi(A))^\sim \not= \emptyset.
$$
By Lemma \ref{lenov}  there exists a subset $W \subset E(0,0)$  such that $W \subset \phi ({\cal O}) \cap (\phi (A))^\sim$ and $W$ is homeomorphic to ${\R}^{n^2 - 2p'q'}$.
Therefore there exists a subset $W_1 = (\phi_{|{\cal O}})^{-1} (W) \subset A^\sim$ that  is homeomorphic to ${\R}^{n^2 - 2p'q'}$. By Lemma \ref{lem5} and the fact that $A \in E(p,q)$ we have
$$
n^2 - 2p'q' \le n^2 - 2pq,
$$
as desired.
\end{proof}

\begin{step}\label{m}
For $1\le  m\le  n-1$, we have either
$$
\phi(E(1, m)) \subset E(1, m)
$$
or
$$
\phi(E(1, m)) \subset E(m, 1). 
$$
Similarly, we have either
$$
\phi(E(m, 1)) \subset E(1, m)
$$
or
$$
\phi(E(m, 1)) \subset E(m, 1). 
$$
\end{step}
\begin{proof}
Let $A\in E(1, m)$ and $\phi(A)\in E(p', q')$. 
By Step \ref{mart}, we have $p'q'=m$. 
We can take an open neighborhood $W\ni \phi(A)$ in $E_n$ such that 
$$
W \subset \bigcup_{0\le  s\le  p',\, 0\le  t\le  q'} E(s, t). 
$$
For any open neighborhood $U\ni A$ in $E_n$ and any integer $0\le  k\le  m$, we can find an element $B\in E(1, k)\cap U$. 
Since $\phi$ is continuous, for any integer $0\le  k\le  m$, we must have $W\cap \phi(E(1, k)) \neq \emptyset$. 
It follows by Step \ref{mart} that 
$$
\{0, 1, 2, \ldots, m\} \subset \{st \ : \ 0\le  s\le  p',\, 0\le  t\le  q'\},  
$$
which implies $\{p', q'\} = \{1, m\}$. 
By the continuity of $\phi$ we obtain either $\phi(E(1, m)) \subset E(1, m)$ or $\phi(E(1, m)) \subset  E(m, 1)$. 
\end{proof}

\begin{step}\label{n-1}
We have either
$$
\phi(E(1, n-1)) = E(1, n-1)
$$
or
$$
\phi(E(1, n-1)) = E(n-1, 1). 
$$
\end{step}
\begin{proof}
By Step \ref{m}, we have either $\phi(E(1, n-1)) \subset E(1, n-1)$ or $\phi(E(1, n-1)) \subset E(n-1, 1)$. 
Note that $E(n-1, 1)$ and $E(1, n-1)$ are compact connected $(2n-2)$-dimensional manifolds without boundaries. 
If $n\ge  3$, then $\phi$ restricted to $E(1, n-1)$ is injective by Step \ref{prvi}, thus the invariance of domain theorem assures the desired conclusion. 
Even when $n=2$, we can see that $\phi|_{E(1, 1)}$ is locally injective and hence locally homeomorphic, which implies that $\phi(E(1, 1)) = E(1, 1)$.
\end{proof}

Combining $\phi$ with the orthocomplementation if necessary, we may and will assume $\phi(E(1, n-1))= E(1, n-1)$ from now on. 

\begin{step}\label{surj}
For any $P\in E(1, n-1)$, we have 
$$
\phi(\{cP + dP^{\perp}\, : \, 0\le  d\le  c\le  1\}) = \{c\phi(P) +d\phi(P)^{\perp}\, : \, 0\le  d\le  c\le  1\}. 
$$ 
Moreover, 
$$
\phi(\{cP \, : \, 0\le  c\le  1\}) \cup \phi(\{P +dP^{\perp} \, : \, 0\le  d\le  1\}) $$ 
is equal to
$$
\{c\phi(P) \, : \, 0\le  c\le  1\} \cup \{\phi(P) + d\phi(P)^{\perp} \, : \, 0\le  d\le  1\}. 
$$ 
We also have 
$$
\phi({\rm Sca} (H)) = {\rm Sca} (H) 
$$
and either 
\begin{equation}\label{01}
\phi(0)=0 \quad {\rm and} \quad \phi(I)=I
\end{equation}
or 
\begin{equation}\label{10}
\phi(0)=I \quad {\rm and} \quad \phi(I)=0.
\end{equation}
\end{step}
\begin{proof}
Let $P\in E(1, n-1)$.
By Lemma \ref{properties}, we see 
$$
\bigcap_{Q\in E(1, n-1)\cap P^{\sim}} Q^{\sim} = \{cP +dP^{\perp}\, : \, 0\le  c, d\le  1\}. 
$$
Since $\phi(E(1, n-1)) = E(1, n-1)$, we have
$$
\phi\left(\bigcap_{Q\in E(1, n-1)\cap P^{\sim}} Q^{\sim}\right) \subset \bigcap_{Q\in E(1, n-1),\,\, \phi(P)\sim\phi(Q)} \phi(Q)^{\sim} =  
\{c\phi(P) +d\phi(P)^{\perp}\, : \, 0\le  c, d\le  1\}
$$
and hence 
\begin{equation}\label{cd}
\phi(\{cP +dP^{\perp}\, : \, 0\le  c, d\le  1\}) \subset 
\{c\phi(P) +d\phi(P)^{\perp}\, : \, 0\le  c, d\le  1\}.
\end{equation}
By Lemma \ref{properties}, we also have $\bigcap_{Q\in E(1, n-1)} Q^{\sim} = {\rm Sca} (H)$ and $\phi({\rm Sca} (H))\subset {\rm Sca} (H)$. 
Moreover, if $A\in E_n$ satisfy $\phi(A)\in {\rm Sca} (H)$, then $\phi(A)^{\sim}\supset \phi(E_n)$ and hence $A^{\sim}=E_n$, which in turn implies $A\in {\rm Sca}(H)$. 
Note that we can identify the right hand side of (\ref{cd}) with a square.
It follows by the continuity of $\phi$ that the image 
$\phi(\{cP + dP^{\perp}\, : \, 0\le   d\le  c\le  1\})\,\, (\ni \phi(P))$ does not `go beyond the diagonal line ${\rm Sca} (H)$ in the square', hence it is a subset of $\{c\phi(P) +d\phi(P)^{\perp}\, : \, 0\le   d\le  c\le  1\}$.

Assume (for a contradiction) that a real number $0 < c_0 < 1$ satisfies $\phi(c_0P) = c_1\phi(P) +d_1\phi(P)^{\perp}$ for some $0\le  c_1< 1$ and $0<d_1\le  1$. 
Since $\phi$ is continuous at $P$, we can take a real number $0<d_0<1$ such that $\phi(P+d_0 P^{\perp}) = c_2\phi(P) + d_2\phi(P)^{\perp}$ with $c_1< c_2\le  1$ and $0\le  d_2<d_1$. 
Since the convex hull of $\{\phi(P+d_0 P^{\perp})\} \cup {\rm Sca} (H)$ contains $\phi(c_0 P)= c_1\phi(P) +d_1\phi(P)^{\perp}$, we obtain $\phi(P+d_0 P^{\perp})^{\sim} \subset \phi(c_0 P)^{\sim}$ by Lemma \ref{convex}. 
Hence we also have $(P+d_0 P^{\perp})^{\sim} \subset (c_0 P)^{\sim}$. 

Let $Q\in E(1, n-1)$ satisfy $Q\neq P$ and $Q\nleq P^{\perp}$. 
Let $t$ be a real number with $0<t<0$. 
We claim the following: 
\begin{itemize}
\item $tQ \sim P+d_0 P^{\perp} \iff tQ\le  P+ d_0 P^{\perp}$.
\item $tQ \sim c_0 P \iff tQ\le  (1-c_0)P + P^{\perp}$. 
\end{itemize}
If these are true,  it is easy to see that the condition $(P+d_0 P^{\perp})^{\sim} \subset (c_0 P)^{\sim}$ can never hold, hence we obtain a contradiction. 
Let us prove the claim. 
By Lemma \ref{lem3}, we have 
$$
(P+d_0 P^{\perp})^{\sim} = \{A+B \ : \ 0\le  A\le  P+d_0 P^{\perp},\,\, 0\le  B\le  (1-d_0)P^{\perp}\},  
$$
hence the rank one effect $tQ$ with $Q\neq P$, $Q\nleq P^{\perp}$ is an element of this set if and only if $tQ\le  P+d_0 P^{\perp}$. 
(Here we use the elementary fact that for $A, B\in E_n$, the matrix $A+B$ is of rank $\le  1$ if and only if $A, B$ are linearly dependent and of rank $\le  1$.)
Thus we obtain the former equivalence, and the latter can be seen similarly. 

Therefore, we obtain 
$$
\phi(\{cP \, : \, 0< c< 1\}) \subset \{c\phi(P) \, : \, 0\le  c\le  1\} \cup \{\phi(P) + d\phi(P)^{\perp} \, : \, 0\le  d\le  1\}. 
$$
In a similar way, we also obtain 
$$
\phi(\{P + dP^{\perp} \, : \, 0< d< 1\}) \subset \{c\phi(P) \, : \, 0\le  c\le  1\} \cup \{\phi(P) + d\phi(P)^{\perp} \, : \, 0\le  d\le  1\}. 
$$
Combine these with the facts 
\begin{itemize}
\item $\phi({\rm Sca} (H))\subset {\rm Sca} (H)$, 
\item $\phi$ is continuous on $E_n$ and injective (by Step \ref{prvi}) on 
$$
\{cP \, : \, 0< c\le  1\} \cup \{P + dP^{\perp} \, : \, 0\le  d< 1\},
$$
and
\item the set $\{cP + dP^{\perp}\, : \, 0\le  d\le  c\le  1\}$ is simply connected and the boundary 
$$
{\rm Sca} (H) \cup \{cP \, : \, 0< c\le  1\} \cup \{P + dP^{\perp} \, : \, 0\le  d< 1\}
$$
is homeomorphic to the circle $S^1$
\end{itemize}
to obtain the desired conclusion. 
\end{proof}

\begin{step}\label{perp}
For any $P\in E(1, n-1)$, we have $\phi(P^{\perp}) = \phi(P)^{\perp}$. 
\end{step}
\begin{proof}
Let $P\in E(1, n)$. 
By (\ref{cd}) and Step \ref{m}, we obtain either $\phi(P^{\perp})= \phi(P)$ or $\phi(P^{\perp})=\phi(P)^{\perp}$. 
We prove that the former option never holds. 
Assume that $\phi(P^{\perp})= \phi(P)$. 
Using the same argument as in the proof of the preceding step, we obtain 
$$
\phi(\{cP +dP^{\perp}\, : \, 0\le  c\le  d\le  1\}) = \{ c\phi(P) +d\phi(P)^{\perp}\, : \, 0\le  d\le  c\le  1\}.  
$$
Let $0\le  c_0<d_0\le 1$. 
We know that $\phi(c_0P +d_0P^{\perp}) \in \{ c\phi(P) +d\phi(P)^{\perp}\, : \, 0\le  d< c\le  1\}$. 
By the preceding step, there exist $0\le  d_1<c_1\le  1$ such that $\phi(c_0P +d_0P^{\perp})= \phi(c_1P +d_1P^{\perp})$. 
By Step \ref{prvi}, we obtain $c_0+c_1=1$ and $d_0+d_1=1$. 
Therefore, the equation $\phi(c_0P +d_0P^{\perp})= \phi((1-c_0)P +(1-d_0)P^{\perp})$ holds for any $0\le  c_0<d_0\le 1$. 
Then the continuity of $\phi$ implies that $\phi(cI)= \phi((1-c)I)$ for any $0\le  c\le  1$, which contradicts (\ref{01}), (\ref{10}). 
\end{proof}

\begin{step}
Let $n=2$. 
Then there exists a unitary or antiunitary operator $U\colon H\to H$ such that either
$$
\phi(A) = UAU^\ast, \quad {A\in E_2}
$$
or
$$
\phi(A) = UA^{\perp}U^\ast, \quad {A\in E_2}.
$$
\end{step}
\begin{proof}
By Steps \ref{prvi}, \ref{n-1} and \ref{perp}, $\phi$ restricts to a homeomorphism on $E(1, 1)$.  
By Steps \ref{prvi} and \ref{surj}, a moment's reflection shows that $\phi$ maps the subset $E_2\setminus {\rm Sca} (H)$ bijectively onto itself. 
Define the mapping $\psi\colon E_2\to E_2$ by $\psi(A)=\phi(A)$ if $A\in E_2\setminus {\rm Sca} (H)$ and $\psi(cI)=cI$ for real $0\le c\le 1$. 
Then $\psi$ is a bijection such that 
$$
A \sim B \iff \psi (A) \sim \psi (B)
$$
for any $A, B\in E_2$. 
By \cite{GeS}, there exists a unitary or antiunitary operator $U\colon H\to H$ such that 
$$
\phi(A) = \psi(A) \in \{UAU^\ast, UA^{\perp}U^\ast\}, \quad A\in E_2\setminus {\rm Sca} (H).
$$
Since $E_2\setminus {\rm Sca} (H)$ is connected and $UAU^\ast\neq UA^{\perp}U^\ast$ on this subset, the continuity of $\phi$ implies the desired conclusion. 
\end{proof}

In what follows, we consider the case $n\ge 3$. 

\begin{step}
If $n\ge 3$, then $\phi|_{E(1, n-1)}\colon E(1, n-1)\to E(1, n-1)$ extends to a standard automorphism on $E_n$. 
\end{step}
\begin{proof}
We know that $\phi|_{E(1, n-1)}$ is a bijection on $E(1, n-1)$. 
For any $P, Q\in E(1, n-1)$ with $P\neq Q$, we have $P\sim Q$ if and only if $P\perp Q$. 
The desired conclusion is now a direct consequence of Uhlhorn's theorem \cite{U}.
\end{proof}

Therefore, without loss of generality, we may assume $\phi(P) = P$ for any projection $P$ of rank one.

For an effect $A$ we denote by $\sigma (A)$ the spectrum of $A$.

\begin{step}\label{avirio}
For every $A \in E(H)$, there exists an injective function $g_A : \sigma (A) \to [0,1]$ such that
$$
\phi (A) = g_A (A) .
$$
\end{step}

\begin{proof}
The desired conclusion is a straightforward consequence of the fact that for every rank one projection $P$ the effect $A$ commutes with $P$ if and only if $\phi (A)$ commutes with $P$.
\end{proof}

\begin{step}\label{ilasbm}
For every $n\times n$ unitary matrix $U$ and every real $r$, $0 < r < 1$, there exists a real $s \in (0,1)$ such that
$$
\phi \left( U \,  \left[ \begin{matrix} {r & 0 & 0 \cr 0 & 0 & 0 \cr 0 & 0 & I_{n-2} \cr} \end{matrix} \right] \, U^\ast \right) = 
 U \,  \left[ \begin{matrix} {s & 0 & 0 \cr 0 & 0 & 0 \cr 0 & 0 & I_{n-2} \cr} \end{matrix} \right] \, U^\ast.
$$
\end{step}

\begin{proof}
Note that 
$$\left\{ U \,  \left[ \begin{matrix} {r & 0 & 0 \cr 0 & 0 & 0 \cr 0 & 0 & I_{n-2} \cr} \end{matrix} \right] \, U^\ast \ : \ 0<r<1\right\} \subset E(n-2, 1).
$$ 
By Steps \ref{m} and \ref{avirio}, the image of this set by $\phi$ is in 
\begin{equation}\label{cube}
\left\{ U \,  \left[ \begin{matrix} {q_1 & 0 & 0 \cr 0 & q_2 & 0 \cr 0 & 0 & q_3 I_{n-2} \cr} \end{matrix} \right] \, U^\ast \ : \ q_1,\ q_2,\ q_3\in [0, 1] \right\} \cap (E(n-2, 1)\cup E(1, n-2)),   
\end{equation}
which is equal to the disjoint union of the four sets 
$$
\left\{ U \,  \left[ \begin{matrix} {t & 0 & 0 \cr 0 & 1 & 0 \cr 0 & 0 & 0_{n-2} \cr} \end{matrix} \right] \, U^\ast \ : \ t \in (0, 1) \right\},\quad \left\{ U \,  \left[ \begin{matrix} {1 & 0 & 0 \cr 0 & t & 0 \cr 0 & 0 & 0_{n-2} \cr} \end{matrix} \right] \, U^\ast \ : \ t \in (0, 1) \right\},
$$
$$
\left\{ U \,  \left[ \begin{matrix} {t & 0 & 0 \cr 0 & 0 & 0 \cr 0 & 0 & I_{n-2} \cr} \end{matrix} \right] \, U^\ast \ : \ t \in (0, 1) \right\},\quad \left\{ U \,  \left[ \begin{matrix} {0 & 0 & 0 \cr 0 & t & 0 \cr 0 & 0 & I_{n-2} \cr} \end{matrix} \right] \, U^\ast \ : \ t \in (0, 1) \right\} 
$$
if $\dim H \ge 4$. 
In this case, by Step \ref{perp}, $\phi$ fixes the projection $U \,  \left[ \begin{matrix} {1 & 0 & 0 \cr 0 & 0 & 0 \cr 0 & 0 & I_{n-2} \cr} \end{matrix} \right] \, U^\ast\in E(n-1, 1)$. 
The continuity of $\phi$ implies that $\phi$ maps the subset
$$
\left\{ U \,  \left[ \begin{matrix} {r & 0 & 0 \cr 0 & 0 & 0 \cr 0 & 0 & I_{n-2} \cr} \end{matrix} \right] \, U^\ast \ : \ 0<r<1\right\}
$$
into itself. 
Similarly, if $\dim H =3$, then the set (\ref{cube}) can be written as a disjoint union of six sets. 
Since $\phi$ fixes every projection in $E(1, 2)\cup E(2, 1)$, we obtain the same conclusion. 
\end{proof}
We fix an $n\times n$ unitary matrix $U$ for a while. 
By Steps \ref{prvi} and \ref{ilasbm}, for every $2 \times 2$ rank one projection $P$ there exists an injective continuous function $f_P : [0,1] \to [0,1]$ with $f_P(1)=1$ such that
$$
\phi \left( U \,  \left[ \begin{matrix} {rP & 0  \cr   0 & I_{n-2} \cr} \end{matrix} \right] \, U^\ast \right) = 
 U \,  \left[ \begin{matrix} {f_P (r) P & 0  \cr  0 & I_{n-2} \cr} \end{matrix} \right] \, U^\ast .
$$
Since $U \,  \left[ \begin{matrix} {0 & 0  \cr   0 & I_{n-2} \cr} \end{matrix} \right] \, U^\ast \in E(n-2, 2)$, by Step \ref{mart} and the continuity of $\phi$ we also see $f_P(0)=0$. 
Thus $f_P : [0,1] \to [0,1]$ is an increasing bijective function.

\begin{step}
We have $f_P(r)=r$ for any $2\times 2$ rank one projection $P$ and real number $r\in [0, 1]$. 
\end{step}
\begin{proof}
We first claim that $f_P = f$ is independent of $P$. Think just of the upper left two by two corners of the matrices treated in the previous paragraphs. We know that for $r, s \in (0,1)$ and
different rank one projections $P$ and $Q$ we have
$$
rP \sim sQ \iff rP + sQ \le I_2.
$$
Let $P$ and $Q$ be any distinct projections of rank one. Find a projection of rank one $R$ such that ${\rm tr}\, (PR) = {\rm tr}\, (QR)$
(the angle between $P$ and $R$ is the same as the angle between $Q$ and $R$). Then Lemma \ref{dirsum} and the above imply 
$$
f_P (r) P + f_R(s) R \le I_2 \iff 
f_Q (r) Q + f_R(s) R \le I_2,
$$
which yields $f_P = f_Q$.

Next, we prove that $f(t) = t$. For every $t \in (1/2, 1)$ we can find rank one projections $P$ and $Q$ with $P\neq Q$ such that
$$
t(P+Q) \le I_2
$$
and 
$$
s(P+Q) \not\le I_2
$$
whenever $s > t$. Hence, by the above argument again, we have 
$$
f(t)(P+Q) \le I_2
$$
and 
$$
f(s)(P+Q) \not\le I_2
$$
whenever $s > t$. 
We obtain $f(t) = t$. 

If $t\in (0, 1/2]$,  we can find $r\in (1/2, 1)$ and rank one projections $P$ and $Q$ with $P\neq Q$ such that
$$
tP+rQ \le I_2
$$
and 
$$
tP+sQ \not\le I_2
$$
whenever $s > r$.
Hence,
$$
f(t)P+rQ \le I_2
$$
and 
$$
f(t)P+sQ \not\le I_2
$$
whenever $s > r$. 
We conclude that $f(t)=t$ for any $t\in [0, 1]$. 
\end{proof}

This implies that for every unitary matrix $U$ and every rank one effect $R\in E_2$ we have
$$
\phi \left( U \,  \left[ \begin{matrix} {R & 0  \cr   0 & I_{n-2} \cr} \end{matrix} \right] \, U^\ast \right) = 
 U \,  \left[ \begin{matrix} {R & 0  \cr  0 & I_{n-2} \cr} \end{matrix} \right] \, U^\ast .
$$
Using \cite[Corollary 2.10]{GeS} we conclude that for every unitary matrix $U$, every $2 \times 2$ nonscalar effect $A$ and every $(n-2) \times (n-2)$ effect $B$ we have either
$$
\phi \left( U \,  \left[ \begin{matrix} {A & 0  \cr   0 & B \cr} \end{matrix} \right] \, U^\ast \right) = 
 U \,  \left[ \begin{matrix} {A & 0  \cr  0 & * \cr} \end{matrix} \right] \, U^\ast ,
$$
or
$$
\phi \left( U \,  \left[ \begin{matrix} {A & 0  \cr   0 & B \cr} \end{matrix} \right] \, U^\ast \right) = 
 U \,  \left[ \begin{matrix} {A^\perp & 0  \cr  0 & * \cr} \end{matrix} \right] \, U^\ast .
$$

Let $A\in E_n\setminus {\rm Sca} (H)$. 
We can find a unitary matrix $U$ and real numbers $0\le  a_1\le  a_2\le \cdots\le  a_n\le  1$ such that $A = U \, {\rm diag} (a_1, a_2, \cdots, a_n) \, U^\ast$. 
By Step \ref{avirio}, we can find real numbers $b_1, b_2, \cdots, b_n$ such that  $\phi(A) = U \, {\rm diag} (b_1, b_2, \cdots, b_n)\,  U^\ast$.  
The above conclusion implies that for any $1\le  j<k\le  n$ with $a_j\neq a_k$, we have either 
$$
a_j=b_j \quad {\rm and }\quad a_k=b_k
$$
or 
$$
a_j=1-b_j \quad {\rm and }\quad a_k=1- b_k.
$$
Using it, we easily see that $\phi(A)\in \{A, A^{\perp}\}$ for any $A\in E_n\setminus{\rm Sca} (H)$. 
Since $\phi$ fixes every element in $E(1, n-1)$, the continuity of $\phi$ leads to our goal $\phi(A)=A$ for any $A\in E_n$.

\end{document}